\documentclass[a4paper]{article}
\usepackage{fullpage}
\usepackage{amsmath, amssymb}
\usepackage{amsthm}
\usepackage{color}
\usepackage{verbatim, xspace}
\usepackage{comment}
\newtheorem{theorem}{Theorem} 
\newtheorem{definition}[theorem]{Definition}
\newtheorem{lemma}[theorem]{Lemma}
\newtheorem{corollary}[theorem]{Corollary}
\newtheorem{fact}[theorem]{Fact}

\newenvironment{proof-sketch}{\noindent{\bf Sketch of Proof}\hspace*{1em}}{\qed\bigskip}
\newenvironment{proof-idea}{\noindent{\bf Proof Idea}\hspace*{1em}}{\qed\bigskip}
\newenvironment{proof-of-lemma}[1]{\noindent{\bf Proof of Lemma #1}\hspace*{1em}}{\qed\bigskip}
\newenvironment{proof-attempt}{\noindent{\bf Proof Attempt}\hspace*{1em}}{\qed\bigskip}

\newcommand{\probfont}[1]{\textnormal{\textsc{#1}}} 

\newcommand{\tpe}{TP$_{{k}}$-Enum}
\newcommand{\tp}{TP$_{{k}}$}
\newcommand{\rtp}{TP$_{{k}}$-Res}
\newcommand{\kds}{$k$-decomposition}
\newcommand{\ktds}{tree pack($k$)}
\newcommand{\kd}{\kds{}}
\newcommand{\pack}{$k$-pack}
\newcommand{\tpack}{$k$-\probfont{Tree-Pack}}
\newcommand{\ktd}{\ktds{}}
\newcommand*{\mmc}[1]{\mathcal{#1}}
\newcommand*{\nmc}[1]{$\mathcal{#1}$}
\newcommand*{\mc}[1]{\ifmmode\mmc{#1}\else\nmc{#1}\fi}
\newcommand{\si}[1]{\textnormal{SI(}#1\textnormal{)}}
\newcommand*{\mms}[1]{\mathscr{#1}}
\newcommand*{\nms}[1]{$\mathscr{#1}$}
\newcommand*{\ms}[1]{\ifmmode\mms{#1}\else\nms{#1}\fi}

\newcommand{\Oh}{\mathcal{O}}

\newcommand{\defproblem}[3]{
  \vspace{2mm}
  \vspace{1mm}
\noindent\fbox{
  \begin{minipage}{0.95\textwidth}
  #1 \\
  {\bf{Input:}} #2  \\
  {\bf{Task:}} #3
  \end{minipage}
  }
  \vspace{2mm}
}

\newcommand{\kpst}{$k$-\textsc{Partial Spanning Tree}\xspace}

\title{Enumerating tuples of spanning trees}
\author{Rahul CS and Micha{\l} W{\l}odarczyk}

\begin{document}

\maketitle

\section{Introduction}
Consider the problem of enumerating all spanning trees in undirected graphs. Let $G=(V,E)$ be an undirected graph. The problem of enumerating all spanning trees has long history and numerous applications in networks and circuit analysis. We use the notation T-Enum to refer to this problem. The work detailed in  \cite{KR95} gives a simple and efficient solution for the problem where the delay between two successive enumerations is polynomial. Another classical problem of the kind is to compute a set $\mathcal S$ of pairwise edge disjoint spanning trees. The number of trees in the set (tree pack) $\mathcal S$ is part of problem definition. We use the notation \tp{} to represent this problem.  

Motivated by the above two problems, we look at the problem of enumerating all tree packs in a  graph with polynomial delay between two succesive enumerations. Let $k$ be the number of spanning trees in a pack. The problem T-Enum corresponds to the case where $k=1$. 

Tutte \cite{WT61} and Nash-Williams \cite{NW61} came up with the following nice characterization \cite[Theorem 1]{TK12} for the existence of such a pack. 
\begin{theorem}
\label{thm:main-characterization}
 A graph contains $k$ pairwise edge disjoint spanning trees if and only if for every partition \mc P of $V(G)$, the graph $G/$\mc P has $k(|$\mc P$-1|)$ edges.
\end{theorem}
This is such a clinical characterization which gives us the following corollary being recursively evolved  from the original version. 
\begin{corollary}
\label{cor:algorithm}
Let $\mc P$ be a nontrivial partition of $V(G)$ where $G[X]$ contains $k$ pairwise edge disjoint spanning trees for each $X$ in $\mc P$.  Then the graph $G$  has $k$ pairwise edge disjoint spanning trees if and only if for every partition \mc P' of $\mc P$, the graph $G/$\mc P' has $k(|$\mc P'$-1|)$ edges.
\end{corollary}

We extend this  characterization to come up with a polynomial time algorithm for the \textit{tree-pack enumeration problem}, formally defined as follows.  

\defproblem{\probfont{Tree-Pack Enumeration}(\tpe{})}
{An undirected graph $G=(V,E)$ and a parameter $k$}
{Output each set $\mathcal S$ of $k$ pairwise edge disjoint spanning trees with polynomial delay between two successive enumerations}

Note that our problem is not about enumerating all sets of edges that constitutes a tree pack of the above form, but to enumerate the packs itself.

\section{Preliminaries}
We use standard definitions and notations from West\cite{DW00}. The graphs we deal with are undirected multigraphs. Given a graph $G=(V,E)$ having $n$ vertices, we use $V(G)$ (or $V$ when the graph under consideration is clear from the context) to represent the vertex set of $G$. Similarly, $E(G)$ (or $E$) represents the edge set. A \textit{component} in a graph is a maximal subgraph that is connected. A \textit{spanning} subgraph is a subgraph whose vertex set is same as that of  the parent graph. Given subsets $E'$ of $E(G)$, we use the notation $G[E']$ to represent the subgraph of $G$ induced by $E'$. Similarly, given $V'\subseteq V(G)$, $G[V']$ represents the  subgraph of $G$ induced by $V'$. We abuse the notion of cycle as we refer to a pair of parallel edges using the same. Given a nontrivial partition $\{X, V\setminus X\}$ ($\phi \neq X \neq V$) of $V$, a \textit{cut} in a graph $G$ is the set of edges in $G$ having one end in $X$ and the other in $V\setminus X$. An \textit{open cut} is a cut where the associated edge set is an empty set.

Let $\mc P$ be a partition of a given set $S$. A \textit{refinement} $\mc Q$ of $\mc P$ is defined as a partition of $S$ where each $Q\in \mc Q$ is a subset of some $P\in \mc P$.

Given a partition $\mathcal P$ of $V(G)$, we use the notation $G/\mathcal P$ to refer to the \textit{quotient} graph obtained by applying vertex contraction in $G$ where vertices in the same part in $\mathcal P$ gets contracted to a single vertex. Further, drop the so formed loop edges if any. Formally, the vertex set of $G/\mathcal P$ is $\mathcal P$. Edge set of $G/\mathcal P$ is formed by starting with an empty graph and adding an edge $(P_u,P_v)$ to $E(G/\mathcal P)$ for each edge $(u,v)$ in $G$, where $u\in P_u$, $v\in P_v$ and $P_u \neq P_v$. Given a partition $\mc P$ of $V(G)$, we use the term \textit{multicycle} to refer to a cycle $C$ in $G$ that is not strictly contained in any part in $\mc P$. Thus, the multicycle contains at least one edge in $G/\mc P$ and induces at least one cycle in $G/\mc P$. 
\begin{fact}
\label{fact:cycle}
 Let $\mc P$ be a partition of $V(G)$. Let $P\in \mc P$ and $C$ be a cycle in $G[P]$. Let $\mc P'$ be a refinement of $\mc P$ such that at least one pair of vertices in $V(C)$ belongs to distinct parts in $\mc P'$. Then, there exists an edge $e$ in $C$ such that the end vertices of $e$ belongs to distinct parts in $\mc P'$.
\end{fact}

A $k$-decomposition $\mathcal T$ of a graph $G$ is a $k$-tuple $(G_1$, $G_2$ ,$\ldots$, $G_k)$ where each $G_i$ is a spanning subgraph, and $\{E(G_i): 1\leq i \leq k\}$ is a partition of $E(G)$. We define a \pack{} of $G$ to be a $k$-decomposition of an arbitrary spanning subgraph of $G$. Note that, the $k$-tuple composed of the edge sets of the coordinates in a \pack{} uniquely identifies the \pack{}. We use the notation $E_i(\mc T)$ to refer to the edge set of the $i^{th}$ coordinate in $\mc T$. We define a \tpack{} in a graph to be a \pack{} where each coordinate is a tree. Further, given a pair of \pack{}s   \mc T, $\mc T'$ of $G$, we say \mc T \textit{restricts} \mc T' (\mc T' \textit{restricted by} \mc T) if $E_i(\mc T)\subseteq E_i(\mc T')$ for $1\leq i \leq k$.
\begin{definition}[Tuple isomorphism]
We call $k$-tuples $(a_1, \dots a_k),\, (b_1, \dots b_k)$
isomorphic if there is a~permutation $\pi \in S_k$ such that $a_i = b_{\pi(i)}$ for each $1 \le i \le k$.
\end{definition}

\section{The Restricted Tree Packing Problem}
In this section, we look at a variant of the \tp{} problem. The input to the problem is a set of pairwise edge disjoint forests $F_1$, $F_2$ ,$\ldots$, $F_k$ of $G$ and the question of interest is to compute a a set of pairwise edge disjoint spanning trees $T_1$, $T_2$,$\ldots$, $T_k$ where $E(F_i)\subseteq E(T_i)$ for $1\leq i\leq k$. We refer to this problem as \rtp{}, formally defined as follows. 

\defproblem{\kpst(\rtp{})}
{Undirected multigraph $G$, and a \pack{} $\mc T=(F_1,F_2,\ldots,F_k)$}
{Compute a \tpack{} $\mc T'=(T_1,T_2,\dots, T_k)$ restricted by $\mc T$}


We solve the \rtp{} problem using the following characterization which is an extension of the characterization in \cite[Theorem 1]{TK12}. We give a constructive proof for the characterization which is similar in spirit to the proof in \cite{TK12}.
\begin{theorem}
\label{thm:characterization}
 Let $G=(V,E)$ be an undirected graph and let $\mc T=(F_1,F_2,\ldots,F_k)$ be a \pack{} of  $G$. Let $E_{f}$ be the set $\displaystyle E\setminus \cup_{i=1}^k E(F_i)$. The graph $G$ has a \tpack{} restricted by \mc T if and only if, for each partition $\mathcal P$ of $V(G)$, $\displaystyle \Sigma_{i=1}^k$ def($F_i/\mathcal P)\leq |E((V,E_f)/\mathcal P)|$. 
\end{theorem}
Note that the characterization makes sense only when each $F_i$ in \mc T is a forest.

To start with, we tweak the algorithm for \tp{}, detailed in \cite{TK12} and its proof in a way to extend later to solve \rtp{}. As in \cite{TK12}, our algorithm is inductive in nature where we assume the solution for the problem to be available for the parameter value being $k-1$ and we solve for $k$. The base case ($k=1$) is straightforward. Further, this modified version is tighter in terms of runtime analysis. We use the label \mc A to refer to the algorithm in \cite {TK12} and we assume we deal with a "yes" instance of the problem.
\subsection{A Reformulation of the Algorithm for \tp{}}
The algorithm \mc A is a tail recursive algorithm which ensures a steady progress in each recursion. The input to each recursion in the algorithm is a \kd{} \mc T and the output being another \kd{} \mc T' which being fed as the input to the next recursion, unless the progress function hits its threshold. Once the algorithm terminates, it outputs either a \ktd{} or a \kd{} certifying the violation of the characterization in Theorem \ref{thm:main-characterization}.

A single recursion of \mc A takes a \kd{} \mc T as input having $k$-1 of its coordinates being trees. The remaining coordinate is a non-tree, as otherwise we are done. There is an arbitrary strict ordering between the coordinates, with the last being the non-tree. Our algorithm follows the same structure with each recursion constitutes of two phases. The first phase is exactly the same as in \mc A. In the second phase, \mc A performs a single edge swap between the coordinates $G_k$ and some $G_j$ in $\mc T$, retaining the connectedness of $G_j$ ($j\neq k$). We modify this part in such a way that we perform minor shortcutting operations on the computation path associated with $\mc A$.

In the first phase, the algorithm performs a sequence of partition refinement operations starting from the partition $\mc P_0=\{V(G)\}$. Each refinement constructs a new partition $\mc P_{i+1}$ from $\mc P_i$ as follows: For each $X$ in $\mc P_i$ and coordinate $G_j$ in \mc T (iterate $j$ from 1 to $k$ as in the predefined ordering), check whether $G_j[X]$ is connected. Pick the least $j$ violating the condition. Now, we construct the partition $\mc P_{i+1}$ from $\mc P_i$ by adding the vertex set of each component in $G_j[X]$ as a part to $\mc P_{i+1}$, for each $X$ in $\mc P_i$. Thus, $\mc P_{i+1}$ is a refinement of $\mc P_i$. At this point, we define an indexing function SI($\mc P_i$) (the \textit{split index} of the partition $\mc P_i$) as the index $j$ of the coordinate in $\mc T$ being responsible for the split. We address this coordinate as the splitter of $\mc P_i$, denoted as $S(\mc P_i)$ which is exactly same as $G_{\si{\mc P_i}}$. Further, the level $l(e)$ of an edge is defined as the largest $i$ such that some part in $\mc P_{i}$ contains both  ends of $e$. 

This refinement procedure terminates at least when $\mc P_i$ is $V(G)$. This final partition in the sequence is being addressed as $\mc P_{\infty}$, where $G_j[X]$ is connected for each part $X$ in $\mc P_{\infty}$ and each $j$ between 1 and $k$. The split index of $\mc P_{\infty}$ is infinity($>k$). Note that $\mc P_1$ is induced by the components in $G_k$. The split index of $\mc P_0$ is always at least $k$. Our objective is to hit a recursion where $\mc P_0=\mc P_{\infty}$ (the split index of $\mc P_0$ is greater than $k$). 

\begin{lemma}
\label{lem:extend}
 Let $P$ be a part in a partition $\mc P_i$ in a sequence constructed above. Let $u$, $v$ be a pair of vertices in $P$ and be in the same component in $S(\mc P_i)$. Then $u$ and $v$ belong to distinct parts in $\mc P_{i+1}$ if and only if there exists an edge $e'$ in each  path from $u$ to $v$ in $S(\mc P_i)$ having one of its ends in $P$ ($P_{uv}$ spans beyond $P$).
\end{lemma}
\begin{proof}
Let us consider the forward direction of the claim. Clearly, there exists at least one path $P_{uv}$ from $u$ to $v$. For a contradiction, let us assume there does not exist an edge $e'$ in $P_{uv}$ having exactly one of its ends in $P$. This implies, all the edges in $P_{uv}$ is are present in $S(\mc P_i)[P]$. This contradicts the premise that $u$ and $v$ belongs to distinct parts in $\mc P_{i+1}$. 

To prove the reverse direction, observe that there does not exist a path $P_{uv}$ in $S(\mc P_i)[P]$. In other words, they belong to distinct components in $S(\mc P_i)[P]$.
\end{proof}

\begin{lemma}
 \label{lem:level}
 Given a partition $\mc P_i$ in the sequence and a $G_j$ in $\mc T$, the level of each edge in $G_j/\mc P_{i}$ is less than $i$. 
\end{lemma}
\begin{proof}
 Observe that an edge $e$ in $G_j/\mc P_i$ has its end vertices belong to distinct parts in $\mc P_i$. Hence the partition $\mc P_{l(e)}$ has to be coarser than $\mc P_i$. 
\end{proof}

\begin{lemma}
 \label{lem:cycle}
Let $\mc P_i$ be partition in the sequence and let $G_j$ in $\mc T$. Let $\mc C$ be a multicycle in $G_j/\mc P_i$ and let $e$ be an edge of least level in $C_j/\mc P_i$. Then, the  vertex set of $\mc C$  is  contained in some part $P\in \mc P_{l(e)}$.
\end{lemma}
\begin{proof}
 We assume the claim to be wrong and get a contraction. Let there be a pair of vertices $u$, $v$ in $V(\mc C)$ such that $u$ and $v$ belongs to distinct parts $P_1$, $P_{2}$ in $\mc P_{l(e)}$. Consider the graph $G_j/\mc P_{l(e)}$. Pick an arbitrary path $P_{uv}$ in $\mc C$. The path $P_{uv}$ induces a path between $P_1$ and $P_2$ in $G_j/\mc P_{l(e)}$. Thus, at least one edge $e'$ in $P_{uv}$ induces an edge in $G_j/\mc P_{l(e)}$ and from Lemma \ref{lem:level}, the level of $e'$ is less than $l(e)$. This gives us the required contradiction.
\end{proof}

Before getting in to the second phase, the algorithm searches for a multicycle  in $G_{rf}/P_{final}$. If there does not exist one, the algorithm terminates reporting failure.
In the second phase, we perform a sequence of \textit{edge exchange} operations between the coordinates of $\mc T$. By doing so, we incorporate  multiple recursions of algorithm \mc A in to the second phase. In our algorithm, the second phase itself is a recursive  subroutine with two input parameters. We refer to this subroutine with the label CONNECTOR($rf$, $final$). The first parameter $rf$ is an index referring to a coordinate in $\mc T$, and $final$  in an index referring to a partition in the partition sequence $\mc P_{0}$, $\mc P_{1}$,$\ldots$, $\mc P_{\infty}$. We assume global access to the \kd{} \mc T and the partition sequence. 
The initial call to this phase being made with parameters $k$ and $\infty$. 

The subroutine picks an arbitrary multicycle $\mc C$ in $G_{rf}/P_{final}$.  
Then the algorithm picks an edge $e$ in $\mc C$ 
with least level $l(e)$ and performs the assignment $i=l(e)$. 
The algorithm moves the edge $e$ from $G_{rf}$ to $S(\mc P_{i})$. This procedure breaks a cycle in $G_{rf}$. If the end vertices of $e$ are already connected in $S(\mc P_{i})$, then we invoke the next recursion in the subroutine, setting $final=i$ and $rf=$SI($\mc P_i$). Otherwise, this is the end of the current recursion in the subroutine and we invoke the next recursion of our algorithm with the parameter being \mc T.


\subsubsection*{Proof of Correctness}
We follow the same line of proof as in \cite{TK12}. The way we modify the algorithm helps us to redefine the progress function associated with the algorithm in a much simpler way.  We claim that in each recursion of our algorithm, there is a strict reduction in the number of components in $G_k$, which is enough to ensure a positive progress towards the solution.

Recall that by the end of first phase, each $X$ in $\mc P_{\infty}$ has the property that $G_{j}[X]$ is connected for each $G_j$ ($1\leq j\leq k$) in $\mc T$ and hence a tree pack($k$) for each $G[X]$. These tree packs together with the solution for $G/\mc P_{\infty}$ constitutes the final solution.  This is by virtue of Corollary \ref{cor:algorithm}. It is easy to see that this construction is recursive in nature. This is an additional filter function that could be added to our algorithm and $\mc A$.

Next we discuss the correctness of the second phase which itself is a recursive subroutine. We make the following claim.
\begin{lemma}
Each call to the subroutine CONNECTOR($rf$,$final$) maintains the invariant that the subgraph $G_{rf}/\mc P_{final}$ has at least one multicycle in it. 
\end{lemma}
\begin{proof}
Consider the first call to the subroutine CONNECTOR from the first phase. Clearly, $G_k$ contains an open cut  as otherwise we are done. Assuming the instance satisfying the characterization, by counting argument, there exists a multicycle in $G_k/\mc P_{\infty}$ as all the remaining coordinates are trees and contributes at most $|\mc P|-1$ edges to $G_k/\mc P_{\infty}$. If not, we have a certificate vindicating that the instance fails the characterization, and we stop. Thus, the algorithm progresses to the second phase for yes instances.

For subsequent calls, we make an inductive argument. Let us assume that the invariant holds for a call to some intermediate recursion. We assume a recursive call follows, and prove the invariant for that. Let us consider the execution sequence associated with this recursion. We pick the multicycle $\mc C$ in $G_{rf}/\mc P_{final}$. Clearly, $final\neq 0$ and $\mc C$ induces at least one edge in $G_{rf}/\mc P_{final}$. From Lemma \ref{lem:level}, the level of each edge $e$ in $\mc C/\mc P_{final}$ is strictly less than $final$. Let $e=(u,v)$ be the edge with least level in $\mc C$ and the algorithm sets $i=l(e)$. From Lemma \ref{lem:cycle}, there exists $P\in \mc P_i$ such that $V(\mc C)\subseteq P$. 
The algorithm moves the  edge $e$ to $S(\mc P_{i})$. From the assumption that this is not the last recursion, $u$ and $v$ are connected by a path $P_{uv}$ in $S(\mc P_i)$ (excluding $e$). From Lemma \ref{lem:extend}, each $P_{u,v}$ in $S(\mc P_i)$ spans beyond $\mc P$ and has at least one edge in $S(\mc P_i)/\mc P_i$. Thus, the addition of $e$ to $S(\mc P_i)$ creates a multicycle in $S(\mc P_i)/\mc P_i$.    
\end{proof}

To conclude, observe that the subroutine CONNECTOR always pulls out the edge $e$ from a cycle and this does not disturb connectedness of any of those coordinates in $\mc T$. Further, it terminates when the end vertices of the edge $e$ belong to different components in $S(\mc P_i)$. This happens when SI($\mc P_i$)$=k$ and $i$ takes the value 0. Thus $final$ always takes a positive value.  Since, the cycle $\mc C$ under consideration is in $G_{rf}/\mc P_{final}$ and completely contained in a part in $\mc P_i$, $i<final$ and there is strict decrease in the value of the second parameter over recursions. Thus, CONNECTOR always terminates and at this point, our algorithm makes a positive progress.

As there is a strict reduction in the number of components in $G_k$ over each recursion of our algorithm, there are at most $\Oh(n)$ recursions. 

\section{Extending the Algorithm for \tp{} to solve \rtp{}}
In this section, we discuss an algorithm for \rtp{} which is based on the algorithm for \tp{} detailed in the previous section. The algorithm and the correctness arguments follow similar structure with modifications to meet  the additional constraints imposed by the problem. We assume that for an edge $e$ with $I(e)\neq 0$ belongs to the coordinate $G_{I(e)}$ in $\mc T$. Of course, if a set of edges with nonzero index values form a cycle, we blindly stop. Further, for those edges with $I(e)$ being zero may belong to any of those possible $k$ coordinates. As before, we assume the first $k-1$ coordinates in $\mc T$ to be trees.

In the first phase, we modify the partition refinement rule as follows. Let $\mc P_i$ be an intermediate partition in the sequence. For each $X$ in $\mc P_i$ and for each $1\leq j \leq k$, check for a pair of vertices  $u$, $v$ in $X$ where $u$ and $v$ are not connected in $G_j$ or each path $P_{uv}$ has at least one flex edge in $G_j/\mc P_i$. If such a pair does not exist, then this is the final partition $\mc P_{\infty}$. Otherwise, we fix the least $j$ violating the condition (as before) and construct the refinement $\mc P_{i+1}$ of $\mc P_{i}$, by considering each part $X$ in $\mc P_i$ one by one and for every $u$, $v$ pair in $X$ violating the condition, we make sure that they go to distinct parts in $\mc P_{i+1}$, otherwise the same. Note that in our previous algorithm, each part $X$ in $\mc P_{\infty}$ is such that $G_j[X]$ is connected for $G_j\in \mc T$, and need not be the case here. 

In the second phase, the first call to the subroutine takes the form RES\_CONNECTOR($k$,$\infty$). Further, the subroutine picks the edge $e$ to be the flex edge with the least level in the multicycle $\mc C$ in $G_{rf}/P_{final}$ and fix the index $i=l(e)$. Now, the algorithm moves the edge $e$ to $S(\mc P_i)$. As before, if $e$ connects a pair of vertices in the same component in $S(\mc P_i)$, then the subroutine makes a recursive call RES\_CONNECTOR(SI($\mc P_i$),$i$). Otherwise the subroutine terminates and the algorithm proceeds to the next recursion if required. 
\subsection{Proof of Correctness}
The progress function remains the same as in the previous algorithm. The invariant is slightly different and as follows.
\begin{lemma}
 Each recursive call to the subroutine RES\_CONNECTOR($rf$,$final$) maintains the invariant that the subgraph $G_{rf}/\mc P_{final}$ has at least one multicycle having a flex edge in it.
\end{lemma}
\begin{proof}
Consider the partitioning rule applied in the first phase of the algorithm. Let us consider the first call to the subroutine RES\_CONNECTOR($k$,$\infty$). Let $G_j$ represent a coordinate in $\mc T$ for $1\leq j<k-1$. Clearly, $G_j/\mc P_{\infty}$ does not contain a multicycle. Let there be a cycle $\mc C$ in $G_j/\mc P_{\infty}$. Clearly, some part $P\in P_{\infty}\cap V(\mc C)$ is such that $G_{j}[P]$ is disconnected and from the partitioning rule, all these edges in $\mc C$ are fixed edges. Thus, each flex edge in $G_j/\mc P_{\infty}$ turns out to be a cut edge. This implies the number of flex edges in each $G_j/\mc P_{\infty}=def(G_j/\mc P_{\infty})$. Assuming $G_k$ has an open cut, by counting argument, there exists at least one multicycle in $G_k/\mc P_{\infty}$ having a flex edge (at least one flex edge that is not a cut edge). Clearly, $i=l(e)$ is strictly less than $\infty$ and SI($\mc P_{i}$)$\neq k$ as $e\in G_k$. Now the algorithm moves the edge $e$ to the tree $S(\mc P_{i})$ and this creates a cycle.

As before, we use induction to argue the correctness of the subroutine over subsequent iterations. Assume that the invariant holds for some intermediate recursion, and this is not the last, hence about to make the next recursive call. This implies, the edge $e=(u,v)$ and $i=l(e)$ satisfies the property that the pair of vertices $u$ and $v$ are connected in $S(\mc P_i)/\mc P_{i}$ by a path that has at least one flex edge  $e'$ in $S(\mc P_i)/\mc P_{i}$. Further, $l(e')<i$ and $e'$ is the candidate for $e$ in the next recursion. Thus, addition of $e$ to $S(\mc P_i)$ followed by a call RES\_CONNECTOR(SI($\mc P_i$),i) satisfies the invariant.
\end{proof}

\section{Decision trees}

\begin{theorem}\label{thm:enum-trees}
Assume there is an algorithm for \kpst with running time $T(n,m,k)$.
Then there is an enumerating algorithm that enlists all (pairwise non-isomorphic) $k$-tuples of pairwise disjoint spanning trees in $G$ with at most $\Oh(mk\cdot T(n,m,k))$ operations per one output.
\end{theorem}
\begin{proof}
Let us fix an ordering on $E$.
We construct a decision tree with nodes representing partial solutions, i.e. $k$-tuples of pairwise disjoint forests in $G$.
The partial solutions at depth $\ell$ can use only edges from $e_1, \dots e_\ell$.
The root of the tree is given by tuple $\emptyset^k$ and has depth 0.
The leaves for which all $F_i$ contain $n-1$ edges correspond to the $k$-tuples that should be returned.

The algorithm performs a DFS starting at the root and for each node at depth $\ell < m$ it tries to augment some of $F_i$ with edge $e_{\ell+1}$.
If the edges in $F_i$ start forming a cycle then such a node is neglected.
The node also gets neglected if the instance of \kpst given by
$E' = \{e_{\ell+1}, \dots, e_m\}$ and $F_1, \dots, F_k$ has a negative answer.
After neglecting a node we omit its whole subtree since we know it does not contain any solutions to be returned.

Consider a node $w$ at depth $\ell$ given by a partial solution $F_1, \dots, F_k$.
For each $i$ such that $F_i \neq \emptyset$ we construct a tree $F'_i = F_i \cup \{e_{\ell+1}\}$.
With other trees unaffected we check if such a partial solution is extendable to some full solution and, if yes, we proceed recursively to the newly created node.
If the $k$-tuple in the node $w$ contains empty forests we pick the smallest $i$ such that $F_i$ is empty.
We set $F'_i = \{e_{\ell+1}\}$ and we move to the new node.
After checking all possibilities of adding $e_{\ell+1}$ in the node $w$,
we proceed to another node at level $\ell + 1$ with all trees unchanged.
This move corresponds to not using $e_{\ell+1}$ at all.

We argue that all $k$-tuples of spanning trees will be finally generated,
one per each orbit of the tuple isomorphism.
Let $ind(T)$ denote the smallest $\ell$ such that $e_\ell \in T$.
For a set $\{T_1, \dots, T_k\}$ of spanning trees we fix an ordering such that $ind(T_1) < ind(T_2) < \dots < ind(T_k)$.
We can see that by augmenting only the empty forest with the smallest index we will generate only this ordering.

To estimate the number of operations between outputting two results,
note that DFS only enters nodes with at least one proper solution in its subtree.
Therefore the path traveled between visiting two leaves consists of at most $2m$ nodes.
At each node we consider at most $k+1$ children and in every such step we need to launch the algorithm for \kpst.
By multiplying these quantities we obtain the postulated running time.
\end{proof}

\begin{theorem}
Let $G$ be a DAG graph with fixed vertices $s,t$.
There is an enumerating algorithm that enlists all (pairwise non-isomorphic) $k$-tuples of pairwise edge-disjoint (or vertex-disjoint) $s,t$-paths with at most $\Oh(m^2n)$ operations per one output.
\end{theorem}
\begin{proof}
The idea of the proof is the same as in Theorem~\ref{thm:enum-trees}.
Each node in the decision tree is associated with a partial solution $P_1,\dots,P_k$, where $P_i$ is a path starting at $s$.
We fix a topological ordering on $E$.

Consider a node $w$ at level $\ell$ with a partial solution $P_1,\dots,P_k$.
We try to augment every path $P_i$ ending at $u$ where $e_{\ell+1} = (u,v)$.
For each such $i$ we create $P'_i$ by adding $e_{\ell+1}$ to $P_i$ and move recursively to the newly created node.
Again, if there are some empty paths in the partial solution and $u=s$ we augment only the one with the smallest index.
We also create a duplicate node in which $e_{\ell+1}$ will not belong to any path.

In order to omit subtrees with no solutions we check for each node if its
partial solution is extendable.
By Menger's Theorem this is equivalent to computing $s,t$-flow
in a graph with edges $\{e_{\ell+1}, \dots, e_m\}$ and $(s,s_i)$ for $s_i$
being the end of the path $P_i$ (note that in the edge-disjoint case these edges might be parallel).
If we consider the vertex-disjoint case, we firstly apply the known reduction from \textsc{Min $s,t$-Vertex Cut} to \textsc{Min $s,t$-Edge Cut}

The argument that every solution modulo isomorphism will be generated is alike
in Theorem~\ref{thm:enum-trees}.
Let $ind(P)$ be the index of the first edge on path $P$ (since the egdes
are ordered topologically this is also the edge with the smallest index).
For a set $\{P_1, \dots, P_k\}$ we fix an ordering such that $ind(P_1) < ind(P_2) < \dots < ind(P_k)$ and only this tuple will be found by the algorithm.

To estimate the number of operations between outputting two results,
consider the path of length at most $2m$ between the corresponding leaves in the decision tree.
Although each node can have $k+1$ children, observe that the answer returned by the max flow algorithm does not depend on the index of a path to which we add the new edge.
Therefore we need to run the $\Oh(nm)$-algorithm for max flow only twice per node, depending on whether we add an edge or not.
The claim follows.
\end{proof}

\bibliographystyle{plain}
\bibliography{enum}

\end{document}